\documentclass[11pt]{article}
\usepackage[latin1]{inputenc}
\usepackage{amsmath, graphicx, natbib, verbatim}
\usepackage{setspace, mathrsfs, amssymb}
\usepackage{amsthm, cases}
\usepackage{caption,subcaption}
\usepackage{adjustbox}
\usepackage{amsmath}
\usepackage{natbib}
\usepackage{verbatim,mathrsfs,amssymb,graphicx}
\usepackage{subcaption,setspace,cases,xcolor}

\usepackage[hmargin=1in,vmargin=1in]{geometry}

\title{Itemwise conditionally independent nonresponse modeling for incomplete multivariate data}

\author{Mauricio Sadinle and Jerome P. Reiter
\\\\
\textsc{Duke University} }

\newtheorem{theorem}{Theorem}

\newtheorem{definition}{Definition}

\newcommand{\mcX}{\mathcal X}
\newcommand{\mcM}{\mathcal M}

\newcommand\indep{\perp\!\!\!\perp}
\begin{document}
\maketitle

\begin{abstract}
We introduce a nonresponse mechanism for multivariate missing data in which each study variable and its nonresponse indicator are conditionally
independent given the remaining variables and their nonresponse
indicators.  This is a nonignorable missingness mechanism, in that nonresponse for any item can depend on values of other items that are themselves missing.
 We show that, under this itemwise conditionally independent nonresponse assumption, one can define
and identify nonparametric saturated classes of joint multivariate
models for the study variables and their missingness indicators.  We also show how to perform sensitivity analysis to violations
of the conditional independence assumptions encoded by this missingness mechanism.  Throughout, we
illustrate the use of this modeling approach with data analyses.  
\end{abstract}


\section{Introduction}

When data are unintentionally missing, for example due to item nonresponse in surveys, analysts formally should  
base inferences on the joint distribution of the study variables and their missingness or nonresponse
indicators \citep{Rubin76}.  However, this distribution is not identifiable from the data alone \cite[see, e.g.][]{Little93,Robins97,DanielsHogan08}.  Analysts therefore have to rely on identifying assumptions that
are generally untestable. 

In this article, we define a nonresponse mechanism that allows for 
practical and general modeling approaches with
incomplete multivariate data.  We say that we have itemwise conditionally independent nonresponse when each study variable is conditionally
independent of its missingness indicator given the remaining study
variables and their missingness indicators.  This differs from 
missing at random \citep{Rubin76,LittleRubin02, Seamanetal13, MealliRubin15}, which technically
requires that the probability of the observed missingness pattern does not depend on unobserved values.
In fact, the itemwise conditionally independent nonresponse assumption encodes a
nonignorable missingness mechanism, since missingness for any variable can conditionally depend on 
unobserved values of the other variables.  We show that this assumption leads to
a class of nonparametric saturated distributions \citep{Robins97},
meaning that under this assumption we can identify a unique joint
distribution of the study variables and their nonresponse indicators
from the distribution of the observed data.  We show how to construct
this distribution for arbitrary types of study variables, illustrating  
with examples that involve categorical and continuous variables. We also discuss and
illustrate how to perform sensitivity analysis to violations of the conditional independence assumptions.  These sensitivity analyses 
are based on nonparametric saturated distributions and do not impose
restrictions on the observed-data distribution, which is a desirable property \citep{Robins97}. 

Itemwise conditionally independent nonresponse modeling adds to existing approaches to handle nonmonotone, nonignorable nonresponse. 
These approaches include the pattern mixture models of
\cite{Little93}, which impose different restrictions on
the (nonidentifiable) conditional distribution of the missing
variables given the observed variables and each missingness pattern; the permutation missingness models of \cite{Robins97}, which
for a specific ordering of the study variables assume that the
probability of observing the $k$th variable can depend on the previous
study variables and the subsequent observed variables; and, the
block-sequential models of \cite{ZhouBCMAR}, which make identifying
assumptions for blocks of study variables and their missingness
indicators.  
Of course, each of these methods encodes different reasons for missingness, and one typically cannot tell from the data alone which is most plausible. 
Some benefits of the itemwise conditionally independent nonresponse assumption, as we shall show, are that (i) it is straightforward to 
interpret and explain to non-experts, (ii) it can be implemented easily for many models, and, (iii) it is readily modified to allow interpretable sensitivity analysis. 

\section{Setup}

We consider $p$ random variables or items $X=(X_1,\ldots,X_p)$ taking values on a sample space $\mcX$. Let $M_j$ be the missingness or nonresponse indicator 
for item $j$, such that $M_j=1$ when item $j$ is missing and $M_j=0$ when it is observed.  Let $M=(M_1,\ldots,M_p)$ take values 
on $\mcM\subseteq \{0,1\}^p$.  An element $m=(m_1,\ldots,m_p)\in \{0,1\}^p$ is called a missingness pattern, which we shall 
sometimes represent as the string $m_1\ldots m_p$.  Given $m\in\mcM$, we define $\bar{m}=1_p-m$ to be the indicator vector of 
observed items, where $1_p$ is a vector of ones of length $p$.  For a missingness pattern $m$ we define $X_{m}=(X_j: m_j=1)$ to
 be the missing variables and $X_{\bar{m}}=(X_j: \bar m_j=1)$ to be the observed variables, which have sample spaces $\mcX_{m}$ 
and $\mcX_{\bar{m}}$, respectively.  We denote $M_{-j}=(M_1,\ldots,M_{j-1},M_{j+1},\ldots,M_p)$, and likewise $X_{-j}$.  Given
 a generic element of the sample space $x\in\mcX$, we define $x_{m}, x_{\bar{m}}$ and $x_{-j}$ similarly as with the random vectors, and likewise for an element $m\in\mcM$.

Let $\mu$ be a dominating measure for the distribution of $X$, and let $\nu$ represent the product measure between $\mu$ and the
 counting measure on $\mcM$. We assume that there is a positive probability of observing all the items simultaneously,
 that is, $0_p\in\mcM$, where $0_p$ is a vector of zeroes of length $p$.  We call the joint distribution of $X$ and $M$ the
 full-data distribution, and use $f$ to represent its density with respect to $\nu$.  We call the distribution 
involving the observed items and the missingness indicators the observed-data distribution, with 
density $f(X_{\bar{m}}=x_{\bar{m}},M=m)=\int_{\mcX_{m}}f(X=x,M=m)\mu(dx_{m})$.  We assume that the subset of
 $\mcX_{\bar{m}}\times \mcM$ where $f(X_{\bar{m}}=x_{\bar{m}},M=m)=0$ has probability zero.  
For any missingness pattern, we call the conditional distribution of the missing study variables given the observed data the missing-data distribution, with density $f(X_{m}=x_{m}\mid X_{\bar{m}}=x_{\bar{m}},M=m)$. 
We note that \cite{DanielsHogan08} refer to this as the extrapolation distribution. 
Finally, we call the distribution of $M$ given $X$ the missing-data or nonresponse mechanism, with density $f(M=m\mid X=x)$. 
When obvious from context we shall henceforth write $f(x,m)$ instead of $f(X=x,M=m)$, and likewise for other expressions.

A fundamental problem of inference with missing data is that the full-data distribution cannot be identified in a nonparametric fashion, which means that this distribution cannot be recovered asymptotically by repeatedly sampling from it.  Modeling assumptions have to be imposed on the full-data distribution for it to be obtainable from the observed-data distribution, which is all we can identify nonparametrically with infinite samples.  These assumptions represent identifiability restrictions which in turn define classes of full-data distributions that have a one-to-one correspondence with the observed-data distributions.  These classes are called nonparametric saturated \citep{Robins97} or nonparametric identified \citep{Vansteelandtetal06,DanielsHogan08}, 
of which the itemwise conditionally independent nonresponse class is a particular example.

\section{Modeling under itemwise conditionally independent nonresponse}

We begin with a formal definition of the itemwise conditionally independent nonresponse mechanism.
\begin{definition}\label{def:IMAR}
The nonresponse occurs in an itemwise conditionally independent fashion when
\begin{equation*}
X_j\indep M_j\mid  X_{-j},M_{-j}; ~~~ \hbox{for all} ~ j=1,\ldots,p.
\end{equation*}
\end{definition}

The conditional independence statements given by this assumption imply that, for each item $X_j$, its true value does not influence the probability of it being missing once we control for the values of the remaining items and missingness indicators.  It is worth noticing that this assumption does not exclude marginal dependencies between $X_j$ and $M_j$.  

We now show how to construct a full-data distribution such that it encodes the itemwise conditionally independent nonresponse assumption and perfectly fits $f(x_{\bar{m}},m)$ for all $(x_{\bar{m}},m)\in\mcX_{\bar{m}}\times\mcM$, i.e., the resulting class of distributions is nonparametric saturated.  For this purpose, we first need to define a partial order among the
 missingness patterns $\{0,1\}^p$ as follows. Given $m=(m_1,\ldots,m_p),m'=(m'_1,\ldots,m'_p)\in\{0,1\}^p$, we say $m\preceq m'$ if $m'_j=1$ for all $j$ such that $m_j=1$, that is, $m\preceq m'$ if $m'$ indicates at least the same missing items as $m$.  If $m\preceq m'$ but $m\neq m'$, we write $m\prec m'$.  For example, with $p=3$, $001\prec 101\prec 111$, but $001\not\prec 110$. 

\begin{theorem}\label{theo:ident}
For each missingness pattern $m\in\mcM\subseteq\{0,1\}^p$, given $f(x_{\bar{m}},m)>0$, let the function $\eta_{m}:\mcX_{\bar{m}}\mapsto \mathbb{R}$ be defined recursively as 
\begin{align*}
\eta_{m}(x_{\bar{m}}) &= \log f(x_{\bar{m}},m) - \log\int_{\mcX_{m}} \exp\left\{\sum_{m'\prec m} \eta_{m'}(x_{\bar{m}'})I(m'\in\mcM)\right\}\mu(dx_{m}).
\end{align*}
Then, 
\begin{equation}\label{eq:g}
g(x,m) = \exp\left\{\sum_{m'\preceq m} \eta_{m'}(x_{\bar{m}'})I(m'\in\mcM)\right\}
\end{equation}
satisfies
$$\int_{\mcX_{m}}g(x,m)\mu(dx_{m})=f(x_{\bar{m}},m),$$ for all $(x,m)\in\mcX\times\mcM$.
\end{theorem}
\begin{proof}
In general, for a pattern $m$, $\eta_{m}$ is not a function of the missing variables $X_{m}$, which justifies the expression
\begin{equation*}
\int_{\mcX_{m}}g(x,m)\mu(dx_{m}) = \exp\left\{\eta_{m}(x_{\bar{m}})\right\}\int_{\mcX_{m}} \exp\left\{\sum_{m'\prec m} \eta_{m'}(x_{\bar{m}'})I(m'\in\mcM)\right\}\mu(dx_{m}),
\end{equation*}
and replacing the expression of $\eta_{m}(x_{\bar{m}})$ completes the proof.  
\end{proof}

Theorem \ref{theo:ident} implies that $\int_{\mcX\times \mcM} g(x,m)\nu(dx\times dm)=1$, and therefore $g$ induces a distribution
 on the sample space $\mcX\times \mcM$. This full-data distribution is nonparametric identified by construction. We now show that it encodes the itemwise conditionally independent nonresponse assumption.  

\begin{theorem}
The missingness mechanism induced by $g$ in Theorem \ref{theo:ident} leads to itemwise conditionally independent nonresponse.
\end{theorem}
\begin{proof}
We denote $m_{(j;1)}$ a missingness pattern with $m_j=1$, and $m_{(j;0)}$ the same pattern except that $m_j=0$. Provided that either $m_{(j;1)}$ or $m_{(j;0)}$ belong to $\mcM$, we need to show that the expression  
\begin{align*}
\text{pr}_g(M_j=1\mid   M_{-j}=m_{-j},X=x) = \frac{g\{x,m_{(j;1)}\}}{g\{x,m_{(j;0)}\}+g\{x,m_{(j;1)}\}}
\end{align*}
does not depend on $X_j$.  Notice that if $m_{(j;0)}\notin\mcM$, then $g\{x,m_{(j;0)}\}=0$ and the result holds.  Similarly, if $m_{(j;1)}\notin\mcM$, then $g\{x,m_{(j;1)}\}=0$ and the result also holds.  Otherwise, clearly $m_{(j;0)}\prec m_{(j;1)}$, and so we can write
\begin{align*}
g\{x,m_{(j;1)}\} 
&= g\{x,m_{(j;0)}\}\exp\left\{\sum_{\substack{m'\preceq m_{(j;1)}\\m'\not\preceq m_{(j;0)}}} \eta_{m'}(x_{\bar{m}'})\right\}.
\end{align*}
Therefore,
\begin{align*}
\hbox{logit} ~ \text{pr}_g(M_j=1\mid   M_{-j}=m_{-j},X=x)&= \sum_{\substack{m'\preceq m_{(j;1)}\\m'\not\preceq m_{(j;0)}}} \eta_{m'}(x_{\bar{m}'}).
\end{align*}
Since $\eta_{m}$ depends on $X_j$ only if $m_j=0$, and a pattern $m$ with $m_j=0$ such that $m\preceq m_{(j;1)}$ necessarily also satisfies $m\preceq m_{(j;0)}$, we conclude that $\text{pr}_g(M_j=1\mid   M_{-j},X)$ is not a function of $X_j$, which holds true for all $j=1,\ldots,p$.
\end{proof}

We refer to the class of distributions obtained from Theorem \ref{theo:ident} as the itemwise conditionally independent nonresponse distributions.  This class is quite 
flexible and leads to a number of important particular cases, as we show in the following sections.  We emphasize 
that the missing-data mechanism induced by $g$ in \eqref{eq:g} is nonignorable, as $g(M=m\mid X=x)$ is a function of all the items for all $m$.  

Theorem \ref{theo:ident} provides a way of constructing an itemwise conditionally independent nonresponse distribution from a given observed-data distribution. 
If one estimates the observed-data distribution using a consistent estimator,
 then applying Theorem \ref{theo:ident} with this estimated distribution results in a consistent estimator of the itemwise conditionally independent nonresponse distribution.  We follow this plug-in approach in the illustrative examples below.

\section{An itemwise conditionally independent nonresponse model for categorical variables}
\subsection{Relation with hierarchical loglinear models for contingency tables}\label{ss:loglin}
If each variable $X_j$ is categorical taking values in $\{1,\ldots,K_j\}$, the sample space $\mcX$ is finite 
with $\prod_{j=1}^p K_j$ elements that can be organized as cells of a contingency table.  We assume that there are no structural zeroes.
  Let $\nu$ represent the counting measure on $\mcX\times \{0,1\}^p$, so that the densities $f$ and $g$ are probability mass functions. 
 In this case, the functions $\eta_{m}$ in \eqref{eq:g} take a finite number of values corresponding to each value of $\mcX_{\bar{m}}$. 
 These terms correspond to interactions between the observed items $X_{\bar{m}}$ and the missingness indicators for the missing 
variables $M_{m}$.  Indeed, in this case \eqref{eq:g} is a hierarchical loglinear model without interactions that 
involve both $X_j$ and $M_j$ for all $j$, and with one $p$-way interaction, say, $\eta^{X_{\bar{m}}M_{m}}_{x_{\bar{m}}m_{m}}$ for 
each nonparametrically identifiable probability $\text{pr}(X_{\bar{m}}=x_{\bar{m}},M=m)$.  Interactions of higher order are not present since these would necessarily involve $X_j$ and $M_j$ for some $j$. 

To fix ideas, we explicitly develop the case when $p=3$.
The $\eta_{m}$ functions in Theorem \ref{theo:ident} can be re-expressed as 
\begin{align*}
\eta_{000}(x_1,x_2,x_3) = & ~\eta_{x_1x_2x_3}^{X_1X_2X_3} + \eta_{x_1x_2}^{X_1X_2} + \eta_{x_1x_3}^{X_1X_3} + \eta_{x_2x_3}^{X_2X_3} + \eta_{x_1}^{X_1}+ \eta_{x_2}^{X_2}+ \eta_{x_3}^{X_3} + \eta,\\
\eta_{001}(x_1,x_2) = & ~\eta_{x_1x_2 1}^{X_1X_2M_3} + \eta_{x_11}^{X_1M_3} + \eta_{x_21}^{X_2M_3} + \eta_{1}^{M_3},\\
\eta_{011}(x_1) = & ~\eta_{x_1 1 1}^{X_1M_2M_3} + \eta_{11}^{M_2M_3},
\end{align*}
$\eta_{111} = \eta_{111}^{M_1M_2M_3}$, and similarly for $\eta_{100}(x_2,x_3),\eta_{010}(x_1,x_3),\eta_{110}(x_3)$, and 
$\eta_{101}(x_2)$.  This leads to a  familiar expression for loglinear models, where each first order term associated with $M_j$ is the coefficient of a dummy variable that equals 1 if $M_j=1$, first order terms associated with $X_j$ are coefficients of dummy variables for $K_j-1$ categories of $X_j$, and interaction terms are coefficients of products of the corresponding dummy variables \citep[see, e.g.,][]{Agresti12}.  Notice that in this model there is a three-way interaction for each nonparametrically identifiable probability $\text{pr}(X_{\bar{m}}=x_{\bar{m}},M=m)$.  For example, if $m=000$, then $\bar{m}=111$, $X_{\bar{m}}=(X_1,X_2,X_3)$, $M_{m}=\emptyset$, and so $\eta^{X_{\bar{m}}M_{m}}_{x_{\bar{m}}m_{m}}=\eta_{x_1x_2x_3}^{X_1X_2X_3}$, which corresponds to $\text{pr}(X_1=x_1,X_2=x_2,X_3=x_3,M_1=0,M_2=0,M_3=0)$; or if $m=011$, then $\bar{m}=100$, $X_{\bar{m}}=X_1$, $M_{m}=(M_2,M_3)$, and so $\eta^{X_{\bar{m}}M_{m}}_{x_{\bar{m}}m_{m}}=\eta_{x_1 1 1}^{X_1M_2M_3}$, which corresponds to $\text{pr}(X_1=x_1,M_1=0,M_2=1,M_3=1)$.  

To illustrate modeling under the itemwise conditionally independent nonresponse assumption, we now present an application of the 3-variable loglinear model on a commonly studied dataset with item nonresponse.

\subsection{The Slovenian plebiscite data revisited}\label{ss:Slov}

Slovenians voted for independence from Yugoslavia in a plebiscite in 1991.  \cite{RubinSternVehovar95} analyzed three questions 
related to this process included in the Slovenian public opinion survey, which was collected during the four weeks prior 
to the plebiscite.  These authors presented an analysis under ignorability of the missing-data mechanism for the following three key questions:
$X_1$: are you in favor of Slovenia's independence? $X_2$: are you in favor of Slovenia's secession from Yugoslavia? $X_3$: will you attend the plebiscite?
We call these the Independence, Secession, and Attendance questions, respectively.
The possible responses to each of these were \textsc{yes}, \textsc{no}, and \textsc{don't know}.  \cite{RubinSternVehovar95} argued that the \textsc{don't know} option can be treated as missing data, and so will we in this section. 

To implement the itemwise conditionally independent nonresponse approach, we estimate the probabilities $\text{pr}(X_{\bar{m}}=x_{\bar{m}},M=m)$, and follow the formulas of Theorem \ref{theo:ident} to obtain the $g$ density for the full-data distribution.  Here we use a Bayesian approach to estimate the observed-data distribution.  The observed data can be organized in a three-way contingency table with cells corresponding to each element of $\{$\textsc{yes}, \textsc{no}, \textsc{don't know}$\}^3$, as presented in \cite{RubinSternVehovar95}.  We follow \cite{RubinSternVehovar95} in treating these data as being a random sample from a multinomial distribution.  Our prior distribution for the cell probabilities is symmetric Dirichlet with parameter $1/27$.  Under this approach we obtain a posterior distribution on the observed-data distribution, and thereby also obtain a posterior distribution on the itemwise conditionally independent nonresponse distribution for the full data, as induced by $g$.  We took 5,000 draws from the posterior distribution of the observed-data distribution, and for each of these we applied the formulas from Theorem \ref{theo:ident} to obtain draws from the posterior distribution of $g$.  From these we can obtain draws of the implied probabilities for the items, pr$(X=x)$, under the itemwise conditionally independent nonresponse assumption.  

The probabilities pr(Attendance = \textsc{no}) and pr(Independence = \textsc{yes}, Attendance = \textsc{yes}) are of particular interest 
not only because they are practically relevant, but also because 
the results of the plebiscite provided the proportion of Slovenians who did not attend the plebiscite, and the proportion 
who attended and voted for independence.  Some authors have used this as a way of validating their modeling 
assumptions \citep[e.g.][]{RubinSternVehovar95,Molenberghsetal01}.  Arguably, however, the usefulness of these frequencies to 
validate any modeling approach is limited, given that the survey was collected during a period of a month in which propaganda 
for independence increased as days approached the plebiscite day, and there is evidence that the proportion of pro-independence 
potential voters increased steadily during that period \citep{Slovenia_exhibition}.  A perhaps more appropriate modeling approach 
would take into account the time when each interviewee responded to the survey, but we do not pursue this here. 
We therefore refer to the plebiscite results to help illustrate differences for estimates based on alternative missing data mechanisms,
and do not use them to judge which posited missingness mechanism led to the best estimates.  

\begin{figure}
\captionsetup{width=1\textwidth}
         \begin{subfigure}{0.32\textwidth}
								\centering
								\caption{}
                 \label{fig:SlovIMAR}\vspace{-.3cm}
                 \includegraphics[width=1\textwidth]{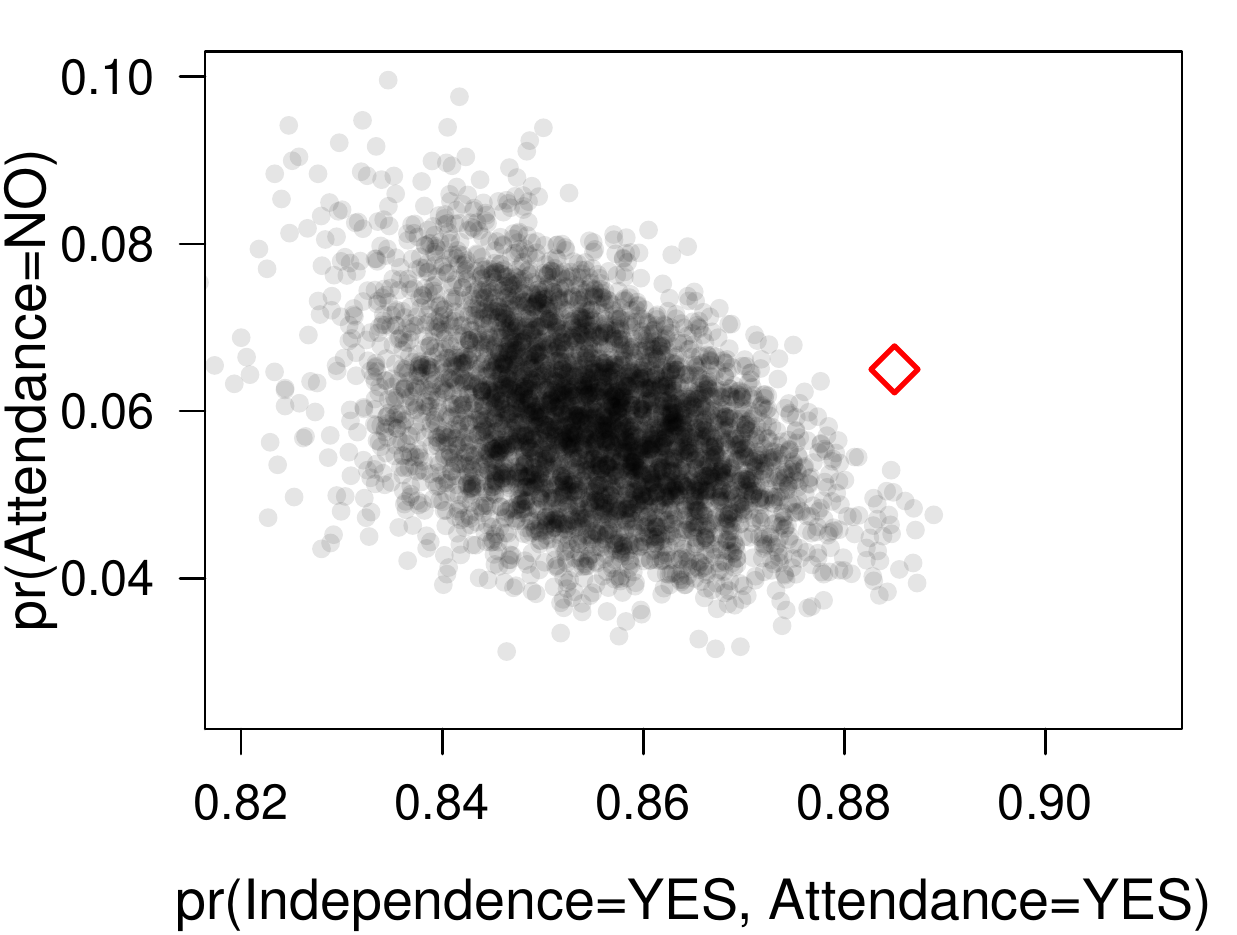}
         \end{subfigure}
         \begin{subfigure}{0.32\textwidth}
                 \centering
								\caption{}
                 \label{fig:SlovMAR}\vspace{-.3cm}
                 \includegraphics[width=1\textwidth]{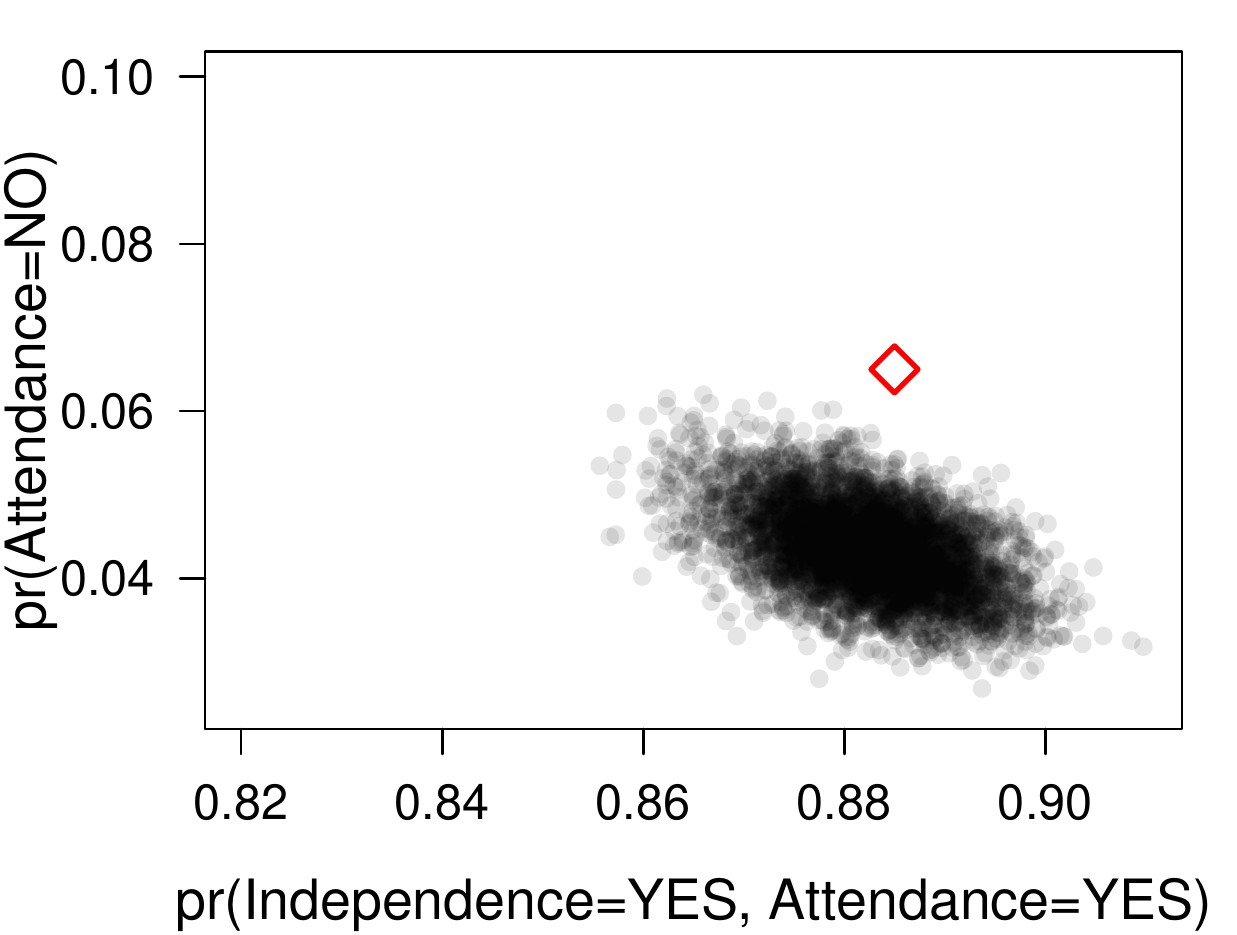}
         \end{subfigure}
     \begin{subfigure}{0.32\textwidth}
             \centering
						\caption{}
             \label{fig:SlovPMM}\vspace{-.3cm}
             \includegraphics[width=1\textwidth]{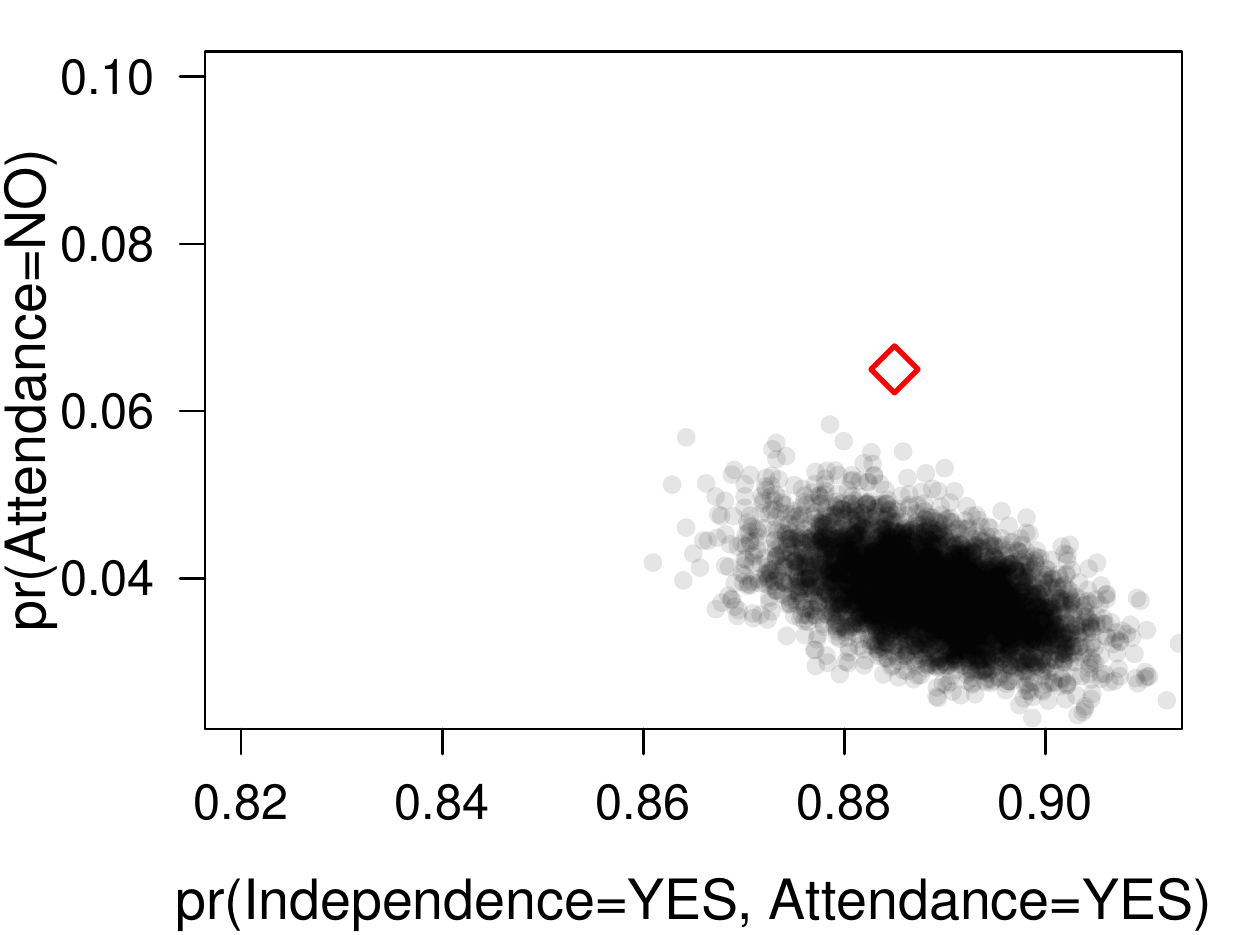}
     \end{subfigure}
\caption{Samples from joint posterior distributions of pr(Independence = \textsc{yes}, Attendance = \textsc{yes}) and pr(Attendance = \textsc{no}) under (a) itemwise conditionally independent nonresponse, (b) an ignorable model, and (c) a pattern mixture model under the complete-case missing-variable restriction of \cite{Little93}.  The plebiscite results are represented by $\color{red}{\boldsymbol \diamond}$.} \label{fig:Slovenia}
\end{figure}

Figure \ref{fig:Slovenia} displays 5,000 draws from the joint posterior distribution of pr(Independence = \textsc{yes}, Attendance = \textsc{yes}) 
and pr(Attendance = \textsc{no}) under itemwise conditionally independent nonresponse, an ignorable missing data model, and a pattern mixture model under the complete-case missing-variable restriction \citep{Little93}.  
None of these approaches produce a joint credible region that covers the plebiscite results, although each approach leads to credible intervals that cover one of the two observed frequencies.  
The key point is that the itemwise conditionally independent nonresponse modeling leads to quite different estimates than the other approaches. 
If using the itemwise conditionally independent nonresponse model returned estimates more similar to those under the ignorable and pattern mixture models, we would have concluded that the 
inferences were not too sensitive to the identifying assumption. In Section \ref{sec:sensitivity}, we perform a sensitivity analysis 
to violations of the itemwise conditionally independent nonresponse assumption for these data.

\section{An itemwise conditionally independent nonresponse model for continuous variables}

\subsection{General modeling strategies}

When the sample space $\mcX=\mathbb{R}^p$, we traditionally assume that the distribution of $X$ is absolutely continuous with respect
 to the Lebesgue measure.  We make the same assumption for the conditional distribution of $X$ given $M=m$, for each missingness pattern $m\in\mcM$, and denote its associated density by $f_m$.  Let $\nu$ represent the product between the Lebesgue and counting measures on $\mathbb{R}^p\times \{0,1\}^p$.  A density $f$ of the joint distribution of $X$ and $M$ with respect to $\nu$ is such that $\int_{\mcX_{m}}f(x,m)dx_{m}=\int_{\mcX_{m}}f_m(x)dx_{m}\text{pr}(M=m)=f_m(x_{\bar{m}})\text{pr}(M=m)$, for all $(x_{\bar{m}},m)\in\mcX_{\bar{m}}\times\mcM$.

In practice we need to specify functional forms for the densities $f_m(x_{\bar{m}})$ based on a sample before using 
the construction given by Theorem \ref{theo:ident}.  A simple option would be to give a parametric form to each $f_m(x_{\bar{m}})$.  
For example, \cite{Little93} proposed to use normal densities in the context of pattern mixture models.  We also can specify each $f_m(x_{\bar{m}})$ in a nonparametric way, for example, using kernel density estimators, provided that we have observations of $X_{\bar{m}}$ given each missingness pattern $m$.  \cite{TitteringtonMill83} followed a similar approach assuming ignorability of the missing-data mechanism.  An analogous approach from a Bayesian point of view would use Dirichlet process mixtures of normals \citep[see, e.g.,][]{EscobarWest95}.  

To fix ideas, we present an example of nonparametric modeling for two variables under the itemwise conditionally independent nonresponse assumption.
When $X=(X_1,X_2)$, $\mcX=\mathbb{R}^2$, and $\mcM=\{00,01,10,11\}$, it is easy to see that Theorem \ref{theo:ident} leads to $g_{00}(x_1,x_2) = f_{00}(x_1,x_2)$,
\begin{align}\label{eq:f01f10}
g_{01}(x_1,x_2) = \frac{f_{00}(x_1,x_2)f_{01}(x_1)}{f_{00}(x_1)}, & ~~~~~ g_{10}(x_1,x_2)  = \frac{f_{00}(x_1,x_2)f_{10}(x_2)}{f_{00}(x_2)},
\end{align}
and
\begin{align}\label{eq:f11}
g_{11}(x_1,x_2) & \propto \frac{f_{00}(x_1,x_2)f_{10}(x_2)f_{01}(x_1)}{f_{00}(x_2)f_{00}(x_1)}.
\end{align}

Hence, by estimating each of the component densities, we derive an itemwise conditionally independent nonresponse full-data distribution that can be applied to data analysis, as we now illustrate.

\subsection{Self-reporting bias in height measurements} 

The National Health and Nutrition Examination Survey is collected in the United States every two years and is
 composed of different modules that include interviews and physical examinations \citep{nhanes}.  In one of the modules, the respondents are asked to self-report their height ($X_1$),
 while in a separate module their actual height is measured by survey staff ($X_2$).  Focusing on these two variables, we can 
informally state the itemwise conditionally independent nonresponse assumption as follows. The association between self-reported height and the reporting of this value 
is explained away by the true height and whether or not this measurement is taken. Similarly, the association between the true height and whether or not this measurement is taken is explained away by the height that would be self-reported and whether or not this value is reported.

We use the combined data from the 1999--2000 and 2001--2002 survey cycles to study the joint distribution of self-reported and
 actual height among individuals who were 18 years or older by the end of year 2000.  Let $w_i$ denote the $i$th sampled 
unit's survey weight for the four year period 1999--2002 so that the U.S. population at the end of year 2000 is the target.
  We estimate the population proportions of each missingness pattern 
as $\hat\pi_{m}=\sum_i w_iI(M_i=m,\text{Age}_i\geq 18)/\sum_i w_iI(\text{Age}_i\geq 18)$ (see Table \ref{t:height}).  
 The estimated proportion of people who would not get their actual height measured given that they would self-report 
their height is $\hat{\pi}_{01}/(\hat{\pi}_{00}+\hat{\pi}_{01})=0.085$, whereas the same proportion among people who would
 not self-report their height is $\hat{\pi}_{11}/(\hat{\pi}_{10}+\hat{\pi}_{11})=0.222$, indicating that there is association
 among the missingness of these two variables.

\begin{table}
\def~{\hphantom{0}}
\caption{Summary measures of the joint distribution of self-reported height ($X_1$) and actual height ($X_2$) given each missingness pattern, under the itemwise conditionally independent nonresponse assumption}\vspace{-6mm}
\begin{center}
\begin{adjustbox}{max width=.98\textwidth}
\begin{tabular}{crcccc}
 \\
Missingness pattern ($m$) & $n_m$~ & $\hat\pi_m$ & $\hat{\text{pr}}_g(X_1>X_2\mid m)$ & $\hat{E}_g(X_1,X_2\mid m)$ & $\hat{\rho}_g(X_1,X_2\mid m)$\\[5pt]
00 & 9,792 & 0.905 & 0.594 &66.8, 66.5 & 0.899\\
01 & 1,059 & 0.084 & 0.575 &66.3, 66.0 & 0.916\\
10 &   235 & 0.009 & 0.614 &64.4, 63.9 & 0.877\\
11 &    54 & 0.002 & 0.587 &63.6, 63.3 & 0.891
\end{tabular}
\end{adjustbox}
\label{t:height}
\end{center}
$n_{m}$, number of observations with missingness pattern $m$; $\hat{\rho}_g$, the estimated correlation; subindex $g$ indicates that these quantities are obtained under the itemwise conditionally independent nonresponse assumption.
\end{table}

\begin{figure}
         \begin{subfigure}{0.32\textwidth}
								\centering
								\caption{}
                 \label{fig:height00}\vspace{-1cm}
                 \includegraphics[width=1\textwidth]{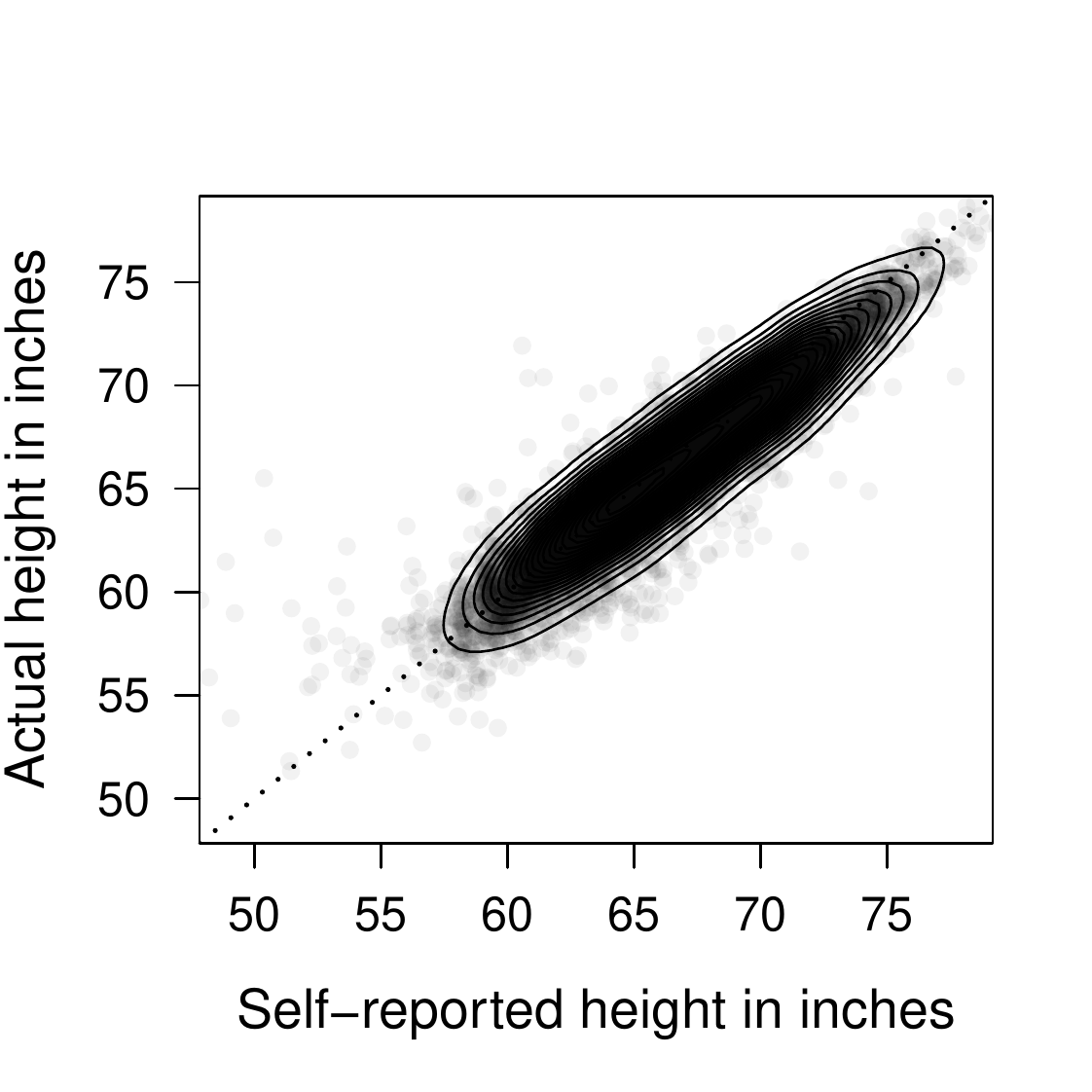}
         \end{subfigure}
         \begin{subfigure}{0.32\textwidth}
                 \centering
								\caption{}
                 \label{fig:height11}\vspace{-1cm}
                 \includegraphics[width=1\textwidth]{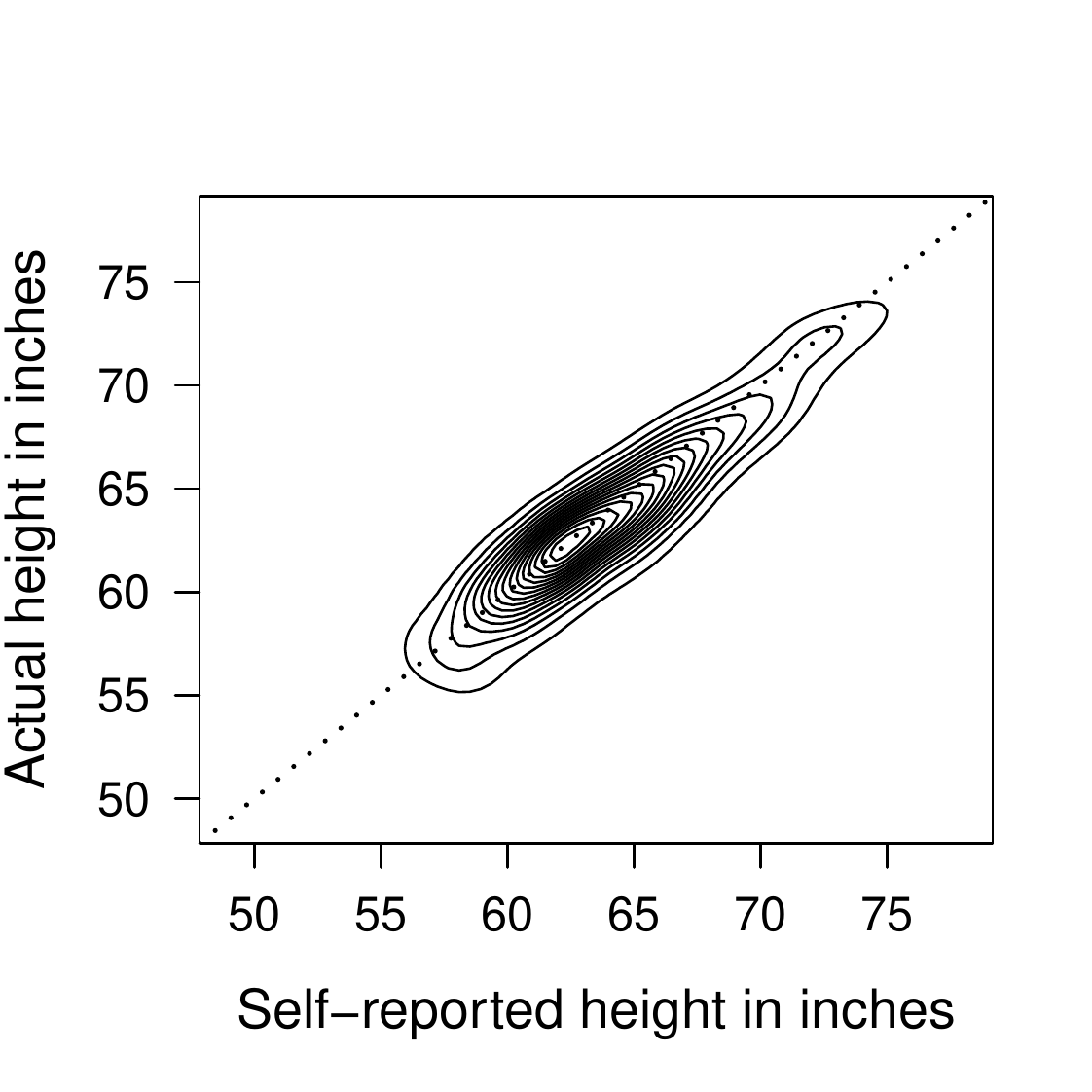}
         \end{subfigure}
     \begin{subfigure}{0.32\textwidth}
             \centering
						\caption{}
             \label{fig:prob_mis}\vspace{-1cm}
             \includegraphics[width=1\textwidth]{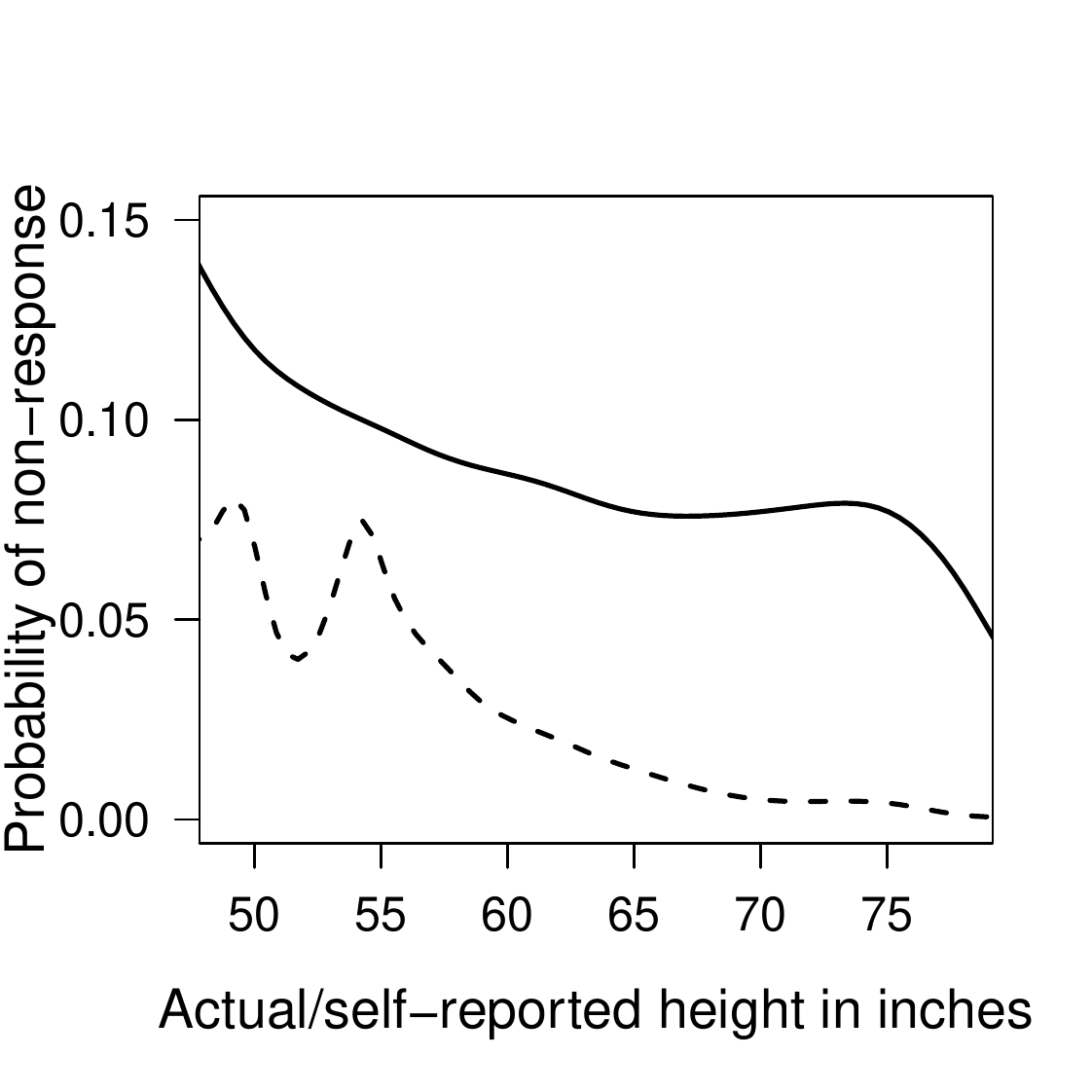}
     \end{subfigure}
\caption{(a) Self-reported versus actual height among respondents who provide both measurements, along with kernel density estimate. (b) Estimated density among individuals who report neither measurement. (c) Estimated probabilities of actual height not being measured given actual height (solid line), and not self-reporting height given the height that would be self-reported (dashed line).  Estimates in (b) and (c) rely on the itemwise conditionally independent nonresponse assumption.} \label{fig:height}
\end{figure}

We estimate each of the nonparametrically identifiable densities $f_{00}(x_1,x_2), f_{10}(x_2)$, and $f_{01}(x_1)$ using kernel density
 estimators with normal kernels, where each kernel component is weighted proportionally to $w_i$, and we choose the bandwidths 
using Silverman's rule  \citep{silverman1986density}.  We obtain the estimated conditional densities $g_m(x)$ by plugging into \eqref{eq:f01f10} and \eqref{eq:f11}.  

Figure \ref{fig:height00} displays the level sets of $\hat{f}_{00}$ along with the self-reported and actual height for individuals for
 which both measurements were recorded.  Figure \ref{fig:height11} displays the estimated density $\hat{g}_{11}$ under itemwise conditionally independent nonresponse.  
The mass of these densities is slightly higher under the 45 degree line, indicating that individuals tend to self-report higher 
values than their actual height.  In Table \ref{t:height} we present the estimated probabilities of self-reported height being larger than the actual height given each missingness pattern, under the itemwise conditionally independent nonresponse assumption, and we can see that this probability is always greater than 0.5.  The density $\hat{g}_{11}$ is also centered around smaller values than $\hat{f}_{00}$ in both dimensions.  In Table \ref{t:height} we show the estimated mean self-reported and true heights for each missingness pattern under itemwise conditionally independent nonresponse, and we can see that the results under this assumption indicate that people who do not report either measurement tend to be shorter than people who do report both measures of height.  

Finally, Figure \ref{fig:prob_mis} displays both the probability of not self-reporting height as a function of the value that would be reported (dashed line), and the probability of the actual height not being measured as a function of its value (solid line).  We can see that as both measures of height become smaller it becomes more likely for both items to be missing.  This illustrates the fact that under the itemwise conditionally independent nonresponse assumption we can capture marginal dependencies between the items and their missingness indicators.

\section{Itemwise conditionally independent nonresponse modeling with monotone missingness patterns}

When a measurement $X_j$ is recorded over $j=1,\ldots,p$ time periods, it is common for dropout or attrition to occur, 
such that once a measurement $X_j$ is not observed nor are $X_{j'}$ for $j'>j$, that is, $M_j=1$ implies $M_{j'}=1$ for 
all $j'>j$.  To use the itemwise conditionally independent nonresponse assumption, we need the probability pr$(M_j=1\mid   M_{-j}=m_{-j},X=x)$ to be defined, and so we require $m_{(j; 1)}\in\mcM$ or $m_{(j; 0)}\in\mcM$, where $m_{(j; 1)}$ is a missingness pattern with $m_j=1$, and $m_{(j; 0)}$ is the same missingness pattern except that $m_j=0$.  In the presence of dropout the only pairs of missingness patterns that have this characteristic are those that correspond to dropout times $j$ and $j+1$.  Letting $T=1+p-\sum_{j=1}^pM_j$ represent the dropout time, $T\in\{1,\ldots,p+1\}$ with $p+1$ representing no dropout, the itemwise conditionally independent nonresponse assumption can be written as pr$(T=j\mid   j\leq T\leq j+1,X=x)=\text{pr}(T=j\mid   j\leq T\leq j+1,X_{-j}=x_{-j})$, or, more naturally, it corresponds to assuming that  the sequential odds pr$(T=j+1\mid   X=x)/\text{pr}(T=j\mid   X=x)$ is not a function of $X_j$.  This assumption is encoded by the itemwise conditionally independent nonresponse distribution, which in this case has a density given by
\begin{equation*}
g(X=x,T=j) = \exp\left\{\sum_{j'\geq j}\eta_{j'}(x_{<j'})\right\}=\exp\left\{\eta_{j}(x_{<j})\right\}g(X=x,T=j+1),
\end{equation*}
where
\begin{align*}
\eta_{j}(x_{<j}) &= \log f(X_{<j}=x_{<j},T=j) - \log\int_{\mcX_{j:p}} \exp\left\{\sum_{j'> j} \eta_{j'}(x_{<j'})\right\}\mu(dx_{j:p}),
\end{align*}
with $X_{<j}=(X_{l}:l<j)$ and $X_{j:p}=(X_{l}:j\leq l\leq p)$.  From this we obtain 
$$\log\frac{\text{pr}_g(T=j+1\mid   X=x)}{\text{pr}_g(T=j\mid   X=x)}=-\eta_{j}(x_{<j}),$$
which means that under this distribution the odds of dropping out at time $j+1$ versus time $j$ only depends on measurements up to time $j-1$.  To the best of our knowledge this assumption has not been used for dealing with monotone nonresponse.

\section{Sensitivity analysis}\label{sec:sensitivity}

\subsection{Exploring departures from the itemwise conditionally independent nonresponse assumption}

One approach for checking how sensitive inferences are to assumptions for handling missing data is to compare results obtained under different approaches, as done for example in Section \ref{ss:Slov}.  An alternative approach, which has been advocated by \cite{Molenberghsetal01} and \cite{DanielsHogan08}, among others, consists in checking the effect of specific parameterized departures from a particular modeling assumption.   In this section we develop this approach for itemwise conditionally independent nonresponse modeling.

Generally speaking, define a sensitivity function as some known function $\xi: \mcX\times\mcM \mapsto \mathbb{R}$.  If for each missingness pattern $m\in\mcM$ the function defined recursively as 
\begin{align}\label{eq:etaxi}
\eta^\xi_{m}(x_{\bar{m}}) &= \log f(x_{\bar{m}},m) - \log\int_{\mcX_{m}} \exp\left\{\sum_{m'\prec m} \eta^\xi_{m'}(x_{\bar{m}'})I(m'\in\mcM)+\xi(x,m)\right\}\mu(dx_{m})
\end{align}
is finite almost surely, we can define 
\begin{equation}\label{eq:gxi}
g_\xi(x,m) = \exp\left\{\sum_{m'\preceq m} \eta^\xi_{m'}(x_{\bar{m}'})I(m'\in\mcM)+\xi(x,m)\right\},
\end{equation}
which would satisfy
$$\int_{\mcX_{m}}g_\xi(x,m)\mu(dx_{m})=f(x_{\bar{m}},m),$$ for all $m\in \mcM$, following the same reasoning as in Theorem \ref{theo:ident}.   This construction is such that the observed-data distribution is constant as a function of $\xi$, the full-data model is identified once $\xi$ is fixed, and the missing-data (extrapolation) distributions are non-constant as a function of $\xi$.  These three properties correspond to the definition of sensitivity parameter given by \cite{DanielsHogan08}.  

Notice that $\xi$ determines the conditional interaction between the $X_j$ and $M_j$ given the remaining variables, as we can see from the log odds ratios
\begin{align}\label{eq:logoddsratio}
\log \frac{g_\xi\{x,m_{(j; 1)}\}/g_\xi\{x,m_{(j; 0)}\}}{g_\xi\{x_{(j;z)},m_{(j; 1)}\}/g_\xi\{x_{(j;z)},m_{(j; 0)}\}}  = & ~ \xi\{x,m_{(j; 1)}\}-\xi\{x,m_{(j; 0)}\}\\
& -\xi\{x_{(j;z)},m_{(j; 1)}\}+\xi\{x_{(j;z)},m_{(j; 0)}\},\nonumber
\end{align}
with $x_{(j;z)}$ being equal to $x$ except that its $j$th entry equals $z$.  However, the exact interpretation 
of $\xi$ is complex and therefore difficult to specify from contextual information or expert opinion.  For example, when the 
variables $X$ are all categorical, $\xi$ determines high order interactions that correspond to functions of odds 
ratios \citep{Bishopetal75}, which are difficult to interpret once we deal with more than three variables, thereby making specifying $\xi$ challenging. 
 Following \cite{DanielsHogan08}, we would like the sensitivity function to be interpretable so that, for instance, it can be specified from contextual information.   The construction given by \eqref{eq:etaxi} and \eqref{eq:gxi} is therefore most useful 
for studying the effect of simple departures from the itemwise conditionally independent nonresponse assumption.  Here, we focus on the set of departures where the odds ratio that measures the dependence between $X_j$ and $M_j$ is constant across the possible values of $X_{-j}$ and $M_{-j}$, that is, $\xi(x,m)=\sum_{j=1}^p\xi_j(x_j,m_j)$.  If we fix $\xi_j(x_j,0)=\xi_j(x_j^*,1)=0$ for all $j$ and for a reference point $(x_1^*,\ldots,x_p^*)\in\mcX$, then $\xi_j(x_j,1)$ corresponds to the log odds ratio of nonresponse when $X_j=x_j$ versus when $X_j=x_j^*$, as in \eqref{eq:logoddsratio}.

\subsection{Sensitivity analysis for the Slovenian plebiscite data}

\cite{RubinSternVehovar95} mention that potential \textsc{no} voters for independence could have been more likely to respond \textsc{don't know} given that not supporting Slovenia's independence was an unpopular position at the time.  If this was the case, then it is possible that the conditional odds of \textsc{don't know} was higher for opponents than for supporters of independence, and not equal as assumed under itemwise conditionally independent nonresponse.  

We explore the effect of assuming that the conditional odds of responding \textsc{don't know} to the independence question 
for \textsc{no} voters was $\exp (\xi_{\text{Ind}})$ times the corresponding odds for \textsc{yes} voters, for 
$\xi_{\text{Ind}}=-5, -1, 0, 1, 5$.  We also explore the effect of fixing the analogous odds ratio 
$\exp(\xi_{\text{Att}})$ for the attendance question, for $\xi_{\text{Att}}=-1, 0, 1$. We keep $\exp (\xi_{\text{Sec}})=1$ for 
the secession question.  Here we take (\textsc{yes},\textsc{yes},\textsc{yes}) as the reference point.  This approach corresponds to augmenting the loglinear model presented in Section \ref{ss:loglin} with the terms $\eta_{\text{\textsc{no}},1}^{X_1M_1}=\xi_{\text{Ind}}$ and $\eta_{\text{\textsc{no}},1}^{X_3M_3}=\xi_{\text{Att}}$.  We follow the same procedure
 described in Section \ref{ss:Slov} to estimate the observed-data distribution, and for each of 5,000 draws from 
its posterior distribution we compute $g_\xi$ as in  \eqref{eq:gxi}, where $\xi(x,m)=\sum_{j=1}^3\xi_j(x_j,m_j)$, with $\xi_1(\text{\textsc{no}},1)=\xi_{\text{Ind}}$, $\xi_3(\text{\textsc{no}},1)=\xi_{\text{Att}}$, and the remaining values of each $\xi_j$ are set equal to zero.  

\begin{figure}
\includegraphics[width=1\textwidth]{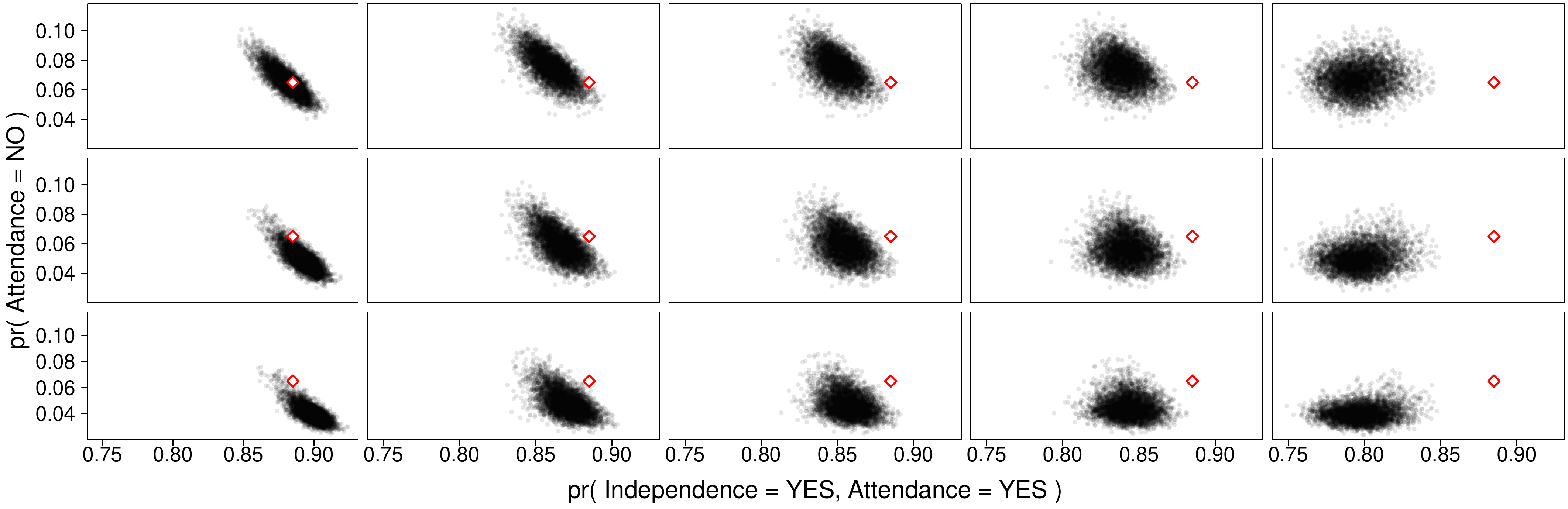}
\caption{Samples from joint posterior distributions of pr(Independence = \textsc{yes}, Attendance = \textsc{yes}) and pr(Attendance = \textsc{no}) under models that depart from the itemwise conditionally independent nonresponse assumption.  The departures are captured by the conditional log odds ratios of nonresponse for the independence question for \textsc{no} versus \textsc{yes}, $\xi_{\text{Ind}} = -5, -1, 0, 1, 5$, from left to right, and for the attendance question for \textsc{no} versus \textsc{yes}, $\xi_{\text{Att}}=-1, 0, 1$, from bottom to top. The plebiscite results are represented by $\color{red}{\boldsymbol \diamond}$.} \label{fig:sensit}
\end{figure}

Figure \ref{fig:sensit} displays the 5,000 draws from the joint posterior of pr(Independence = \textsc{yes}, Attendance = \textsc{yes}) 
and pr(Attendance = \textsc{no}) under each configuration of $g_\xi$.  In this figure, the columns of panels correspond to 
$\xi_{\text{Ind}}=-5, -1, 0, 1, 5$ from left to right, and the rows to $\xi_{\text{Att}}=-1, 0, 1$ from bottom to top; 
the central panel is the same as Figure \ref{fig:SlovIMAR}.  Positive values of $\xi_{\text{Ind}}$ correspond to positive 
conditional association between being opponent to independence and responding \textsc{don't know} to this question.  As $\xi_{\text{Ind}}$ increases, the posterior distribution of pr(Independence = \textsc{yes}, Attendance = \textsc{yes}) gets farther 
from the plebiscite result. Treating the plebiscite results as the relevant true parameter values for illustrative purposes, we would conclude that
 $\xi_{\text{Ind}}=-5$ and $\xi_{\text{Att}}=1$ are reasonable values, suggesting that the itemwise conditionally independent nonresponse assumption is not appropriate for 
these data.  This would indicate that the conditional odds of responding \textsc{don't know} to the independence question for
 \textsc{no} voters was around 0.007 times the corresponding odds for \textsc{yes} voters, 
and the conditional odds of responding \textsc{don't know} to the attendance question for non-attendants was around 2.718 times the 
corresponding odds for attendants.  A plausible interpretation is that potential \textsc{no} voters for independence were more assertive 
with their positions compared to \textsc{yes} voters, while perhaps potential non-attendants were more likely to respond \textsc{don't know} given
 that not being involved in the plebiscite process was unpopular.  Of course, these interpretations are all for illustrative purposes, 
as arguably (see Section \ref{ss:Slov}) the plebiscite results are not the appropriate benchmark given the time difference.
Of course, in practice one does not have any notion of ground truth, and the sensitivity analysis proceeds by examining multiple, plausible
values of the sensitivity parameters to examine differences in the results.  

\section{A word of caution on some related modeling assumptions}

A number of models related to ours have been proposed for dealing with nonignorable missing categorical data \citep[e.g.][]{Fay86,BakerLaird88,Stasny88,Bakeretal92,ParkBrown94}.  Generally speaking, one can consider loglinear models for the $2p$-way contingency table obtained from cross-classifying the study variables and their missingness indicators.  These models can allow each missingness indicator to depend directly on the study variable itself by imposing other constraints. In the literature, the main guidance about the identifiability 
of such models is that they are not identifiable when the number of model parameters exceeds the count of distinct observed cells (possibly plus any 
other observed information, such as supplementary marginal counts).  However, a saturated model does not guarantee a perfect fit \citep{Fay86,BakerLaird88}.  On the other hand, the loglinear model that encodes the itemwise conditionally independent nonresponse assumption does not include interactions between each study variable and its missingness indicator, but it is always identifiable given the result of Theorem \ref{theo:ident}.  It is therefore reasonable to ask: when can we obtain a nonparametric saturated model when allowing $M_j$ to depend on $X_j$ conditioning on the remaining variables in exchange of assuming $X_k\indep M_j\mid X_{-k},M_{-j}$, for some $k\neq j$?

Assuming $X_k\indep M_j\mid X_{-k}, M_{-j}$, we have that $\text{pr}(x_k\mid x_{-k},M_j=1,M_{-j}=0_{p-1})=\text{pr}(x_k\mid x_{-k},M=0_{p})$. Using the law of total probability it is easy to see that 
\begin{align}\label{cd:convex}
\text{pr}(x_k\mid x_{-\{k,j\}},M_j=1,M_{-j}=0_{p-1}) & = \sum_{l=1}^{K_j} \text{pr}(X_k=x_k\mid X_j=l, X_{-\{k,j\}}=x_{-\{k,j\}},M=0_{p}) C_l,
\end{align}
where $C_l=\text{pr}(X_j=l\mid x_{-\{k,j\}},M_j=1,M_{-j}=0_{p-1})$, and so $\sum_{l=1}^{K_j}C_l=1$.  This means that a necessary condition for $X_k\indep M_j\mid X_{-k},M_{-j}$ to hold true is that the probabilities $\text{pr}(x_k\mid x_{-\{k,j\}},M_j=1,M_{-j}=0_{p-1})$ can be written as a convex combination of $\{\text{pr}(X_k=x_k\mid X_j=l, X_{-\{k,j\}}=x_{-\{k,j\}},M=0_{p})\}_{l=1}^{K_j}$.  This condition can be checked using the observed-data distribution because all of these probabilities, except the $C_l$'s, are identifiable.  In other words, we cannot always guarantee a nonparametric saturated model when assuming $X_k\indep M_j\mid X_{-k},M_{-j}$.  

As an example, \cite{Hiranoetal01} consider the case of two variables $X_1$ and $X_2$ where the latter is subject to missingness, 
and they state that the models corresponding to the assumptions  $X_2\indep M_2\mid X_1$ (missing at random) and 
$X_1\indep M_2\mid X_2$ (which they refer to as Hausman--Wise after \cite{HausmanWise79}) cannot be ruled out based on 
the observed data alone.  While this is true for the missing at random model, the Hausman--Wise model corresponds to the assumption presented in the previous 
paragraph. Therefore, it could be rejected in certain situations from the observed-data distribution alone.  
Furthermore, when the number of categories in $X_1$ and $X_2$ differ, the number of constraints imposed by $X_2\indep M_2\mid X_1$ and $X_1\indep M_2\mid X_2$ also differ; 
the assumption in the Hausman--Wise model may correspond to a nonidentifiable model or to one that imposes constraints on the observed-data distribution.  \cite{Hiranoetal01} study 
in detail the case when $X_1$ and $X_2$ are binary and derive closed-form expressions for the full-data distribution under 
$X_1\indep M_2\mid X_2$. Their formulas are not defined when $X_1$ and $X_2$ are independent given $M_2=0$, and result in 
negative estimated probabilities when the condition given by \eqref{cd:convex} does not hold.  These and other related issues had been 
pointed out by \cite{Fay86} and \cite{BakerLaird88}.  It is reasonable to expect that similar complications may arise in more general settings.  On the other hand, these issues do not
 arise under the itemwise conditionally independent nonresponse assumption, which provides an approach that always guarantees a nonparametric saturated model.

\section*{Acknowledgement}

This research was supported by the U.S.A. National Science Foundation via the NSF-Census Research Network.  The first author is also affiliated with the National Institute of Statistical Sciences, Research Triangle Park, North Carolina 27709, U.S.A.

\bibliographystyle{apalike}
\bibliography{biblio_bka}

\end{document}